\documentclass[11pt]{article} 

\usepackage[a4paper,margin=1in]{geometry} 
\usepackage[T1]{fontenc} 
\usepackage[utf8]{inputenc} 
\usepackage{lmodern} 
\usepackage{microtype} 
\usepackage{amsmath,amssymb,amsthm,mathtools} 
\usepackage{booktabs} 
\usepackage{graphicx} 
\usepackage{placeins} 
\usepackage{enumitem} 
\usepackage[round,authoryear,longnamesfirst]{natbib} 
\usepackage{hyperref} 

\hypersetup{hidelinks} 

\newtheorem{theorem}{Theorem} 
\newtheorem{lemma}[theorem]{Lemma} 
\newtheorem{proposition}[theorem]{Proposition} 
\newtheorem{corollary}[theorem]{Corollary} 
\newtheorem{remark}[theorem]{Remark} 
\newtheorem{example}[theorem]{Example} 
\newtheorem{assumption}[theorem]{Assumption} 

\title{Bridging Conformal Prediction and Scenario Optimization:\\
 Discarded Constraints and Modular Risk Allocation} 
\author{Giuseppe C. Calafiore\thanks{Corresponding author.
Tel.: +39 011 090 7071; 
e-mail: giuseppe.calafiore@polito.it.}\\ 
Department of Electronics and Telecommunications\\ 
Politecnico di Torino\\ 
Corso Duca degli Abruzzi 24, 10129 Torino, Italy} 
\date{} 

\begin{document} 
\maketitle 

\begin{abstract}
Scenario optimization and conformal prediction share a common goal, that is,  turning finite samples into safety margins.
Yet, different terminology often obscures the connection between their respective guarantees.
This paper revisits that connection directly from a systems-and-control viewpoint.
Building on the recent conformal/scenario bridge of \citet{OSullivanRomaoMargellos2026}, we extend the forward direction to feasible sample-and-discard scenario algorithms.
Specifically, if the final decision is determined by a stable subset of the retained sampled constraints, the classical mean violation law admits a direct exchangeability-based derivation.
In this view, discarded samples naturally appear as admissible exceptions.
We also introduce a simple modular composition rule that combines several blockwise calibration certificates into a single joint guarantee.
This rule proves particularly useful in multi-output prediction and finite-horizon control, where engineers must distribute risk across coordinates, constraints, or prediction steps.
Finally, we provide numerical illustrations using a calibrated multi-step tube around an identified predictor.
These examples compare alternative stage-wise risk allocations and highlight the resulting performance and safety trade-offs in a standard constraint-tightening problem.
\end{abstract}

\noindent\textbf{Keywords:} conformal prediction; scenario optimization; data-driven control; stochastic model predictive control; system identification;
discarded constraints

\section{Introduction}

Finite-data methods are becoming indispensable in systems and control for turning raw data into explicit safety margins.
The scenario approach is a classical staple in randomized robust design, scenario-based model predictive control (MPC), and data-driven decision making for uncertain systems \citep{CalafioreCampi2005,CalafioreCampi2006,CampiGaratti2008,CampiGarattiPrandini2009,CalafioreFagiano2013_TAC_ScenarioMPC,CalafioreFagiano2013_Automatica_LPVScenarioMPC,SchildbachFagianoFreiMorari2014}.
In parallel, conformal prediction has rapidly gained traction as a model-agnostic calibration layer.
It converts nominal forecasts into distribution-free prediction sets and sees growing use in motion planning and predictive control with learned trajectory models \citep{ShaferVovk2008,LeiGSellRinaldoTibshiraniWasserman2018,AngelopoulosBates2023,DixitLindemannWeiCleavelandPappasBurdick2023,YuZhaoYinLindemann2026}.
In both domains, the workflow is fundamentally the same: the engineer starts from sampled data, and the  goal is a finite-sample statement about future violations or prediction errors, ideally without relying on a fully specified probabilistic model.
From a control perspective, the two theories look strikingly similar.
Both produce data-driven margins ready to be inserted into optimization, monitoring, or safety-filtering loops. Yet, they speak entirely different languages.
Scenario optimization deals in support constraints, violation probability, and discarded scenarios.
Conformal prediction, on the other hand, speaks of nonconformity scores, calibration, and coverage.
The core mechanism behind both theories is exchangeability. If data points are interchangeable, a novel data point only matters if it changes the compressed representation of the final decision rule, or if it falls among explicitly allowed exceptions.
Earlier control-oriented work already linked scenario methods with compression learning, see \citep{MargellosPrandiniLygeros2015}.
More recently, \citet{OSullivanRomaoMargellos2026} established a direct bridge between conformal prediction and scenario optimization.
Their work already maps both directions of that bridge. In one direction, it starts from vanilla conformal prediction and recovers the classical mean-violation law for fully supported scenario programs without discarding.
In the reverse, it interprets conformal quantile selection as a special one-dimensional scenario program with discarded constraints, recovering standard conformal guarantees from that representation.
The contribution of this present note is more specific. Relative to \citet{OSullivanRomaoMargellos2026}, we add a \emph{forward} conformal-style interpretation for general feasible sample-and-discard design algorithms, which are the exact kind used routinely in controller design and scenario MPC.
Furthermore, we provide a simple \emph{modular} rule for combining several blockwise calibration certificates into a single system-level guarantee.
The value of the first contribution is primarily interpretive rather than mathematically stronger than classical scenario theory: it explains, in conformal language, why discarded scenarios enter the finite-sample argument exactly as admissible exceptions.
The value of the second contribution, however, is highly operational.
It provides a transparent risk-budgeting rule for multi-output predictors, stage-wise tubes, and similar structured control problems.
It is equally important to state what this note does \emph{not} claim.
We do not derive new exact tail laws for discarded convex scenario programs;
obtaining sharp results of that nature requires stronger structural assumptions and is well-covered in the existing scenario literature \citep{RomaoMargellosPapachristodoulou2021,RomaoPapachristodoulouMargellos2023}.
Nor do we circumvent the impossibility of generic conditional guarantees based solely on observed support information \citep{GarattiCampi2023Conditional}.
The central guarantee proved here is an expectation bound under a stable retained-reconstruction assumption.
Deriving sharp PAC-style tail bounds at this same level of generality remains outside our scope.
The formal setting throughout is exchangeable data. This perfectly covers the main use cases considered here: i.i.d.
scenario samples, repeated rollouts, and calibration sets collected from independent operating points.
While online conformal methods for dependent time series are certainly important for control, they lie outside the present scope.
To summarize, our contributions are fourfold:
(i) a bridge lemma isolates the exchangeability mechanism behind reconstruction with admissible exceptions;
(ii) a forward bridge extends the no-discarding interpretation of \citet{OSullivanRomaoMargellos2026} to feasible sample-and-discard algorithms with a stable retained reconstruction set;
(iii) a modular calibration rule combines blockwise certificates into a joint guarantee, yielding explicit risk budgeting across outputs or horizon steps;
and
(iv) numerical illustrations demonstrate how different stage-wise allocations of the same total risk budget induce different calibrated tubes and different control trade-offs.

\section{Background and problem statement} 

\subsection{Conformal prediction, calibration, and conditional risk} 

A random sample $Z_1,\dots,Z_m$ is called \emph{exchangeable} if its joint distribution is unchanged by any permutation of the indices.
Independent and identically distributed samples are the main special case of exchangeable samples.
In conformal prediction, one starts from a nominal model and a scalar score that measures how ``unusual'' a point looks relative to that model.
This score is called a \emph{nonconformity score}. 
A \emph{calibration sample} is a held-out exchangeable sample used only to rank these scores and choose a threshold.
Let $Z$ be a test point and let $S=(Z_1,\dots,Z_m)$ be an exchangeable calibration sample.
A set-valued predictor is a measurable map $\Gamma(S)$ that returns a subset of the space of possible test points.
Its conditional violation probability is 
\[
V(S):=\mathbb{P}\{Z\notin \Gamma(S)\mid S\}.
\]
The most common conformal guarantee is marginal validity, namely $\mathbb{P}\{Z\in \Gamma(S)\}\ge 1-\alpha$.
In this paper we also use a PAC-type statement on the random conditional risk, 
\[
\mathbb{P}_S\{V(S)\le \epsilon\}\ge 1-\delta,
\]
which we call a \emph{calibration-conditional certificate}.
This is a probability statement over the random calibration sample $S$: with confidence at least $1-\delta$, the conditional risk after calibration is at most $\epsilon$.
Related notions of conditional validity are known to be subtler than marginal validity \citep{Vovk2013}.
In the basic split-conformal construction, one first fits a nominal predictor on training data independent of the calibration sample, then computes nonconformity scores on the calibration data, sorts them, and selects an order statistic.
For control problems this calibration step can be read as a direct data-driven selection of a safety margin from held-out prediction errors \citep{LeiGSellRinaldoTibshiraniWasserman2018,AngelopoulosBates2023}.
For example, if an identified model predicts a scalar output as $\hat y=f_{\mathrm{id}}(\xi)$ from a regressor $\xi$, then a natural score is the absolute residual $R=|y-\hat y|$.
Choosing a threshold $q$ from the ordered calibration residuals yields the calibrated interval $[\hat y-q,\hat y+q]$.

\subsection{Scenario optimization with discarded constraints} 

Let $S=(\delta_1,\dots,\delta_m)$ be an i.i.d.
sample of scenarios drawn from a probability measure on a space $\Delta$.
A sample-and-discard algorithm chooses a set of discarded indices $D_r(S)\subseteq [m]$, with $|D_r(S)|=r$, and returns a decision $x_r(S)$ that satisfies all retained constraints, 
\[
g(x_r(S),\delta_i)\le 0,\qquad i\notin D_r(S).
\]
Its out-of-sample violation probability is 
\[
V_r(S):=\mathbb{P}\{g(x_r(S),\delta)>0\mid S\},
\]
where $\delta$ is an independent fresh scenario.
Convex scenario programs are the area in which
tail bounds for $V_r(S)$ have been developed first, see \cite{Calafiore2010,CampiGaratti2011}.
In control, this framework  appears in sampled robust design and in finite-horizon scenario MPC formulations \citep{CalafioreCampi2006,CalafioreFagiano2013_TAC_ScenarioMPC,SchildbachFagianoFreiMorari2014}.
The result below needs only feasibility and a suitable reconstruction property, not convexity itself.
The role of support constraints is played here by a retained reconstruction set, namely a subset of retained sampled constraints that already determines the final decision.

\begin{remark}\rm
A retained reconstruction set can be understood as follows: among all retained sampled constraints, only a few are actually needed to pin down the final decision.
Stability means that removing any other retained sample does not change that decision.
This is the same compression mechanism that often sits behind classical scenario bounds.
\end{remark}

\begin{assumption}[Stable retained reconstruction set]\label{ass:stable-support}
For each sample $S$, the algorithm returns measurable, permutation-equivariant maps $D_r(S)\subseteq [m]$ and $C_r(S)\subseteq [m]\setminus D_r(S)$ such that: 
\begin{enumerate}[label=(\alph*),leftmargin=1.8em]
\item the decision $x_r(S)$ satisfies all retained constraints, namely $g(x_r(S),\delta_i)\le 0$ for every $i\notin D_r(S)$;
\item if $j\notin D_r(S)\cup C_r(S)$, then removing the $j$th sample does not change the decision, 
\[
x_r(S^{-j})=x_r(S),
\]
where $S^{-j}$ is the sample obtained from $S$ by removing $\delta_j$.
\end{enumerate}
\end{assumption}

Permutation-equivariant means that if the sample order is relabeled, the selected discarded and reconstruction indices are relabeled in the same way.
Assumption~\ref{ass:stable-support} is not automatic for arbitrary removal heuristics.
It is satisfied, however, by many retained convex scenario programs with a unique optimizer and measurable tie-breaking, in which case $C_r(S)$ may be chosen as any measurable support or compression set of the retained problem.
The set need not be unique; any measurable, permutation-equivariant selection is enough.
If its size is bounded by a constant $\zeta$, then the results below give a bound in terms of $r+\zeta$.
In the fully supported convex case one typically has $\zeta=d$, the number of decision variables \citep{CampiGaratti2018}.

\section{A bridge lemma: reconstruction with admissible exceptions} 

The next lemma is the key structural step.
It says that a future point can be missed only if, in the augmented sample, that point becomes either one of the samples needed to reconstruct the predictor or one of the samples that are explicitly allowed to be exceptional.

\begin{lemma}[Exchangeable reconstruction with admissible exceptions]\label{lem:bridge}
Let $Z_1,\dots,Z_{m+1}$ be exchangeable random elements of a measurable space $\mathcal{Z}$.
For each $n\ge 1$, let 
\[
B_n(Z_{1:n})\subseteq [n],\qquad E_n(Z_{1:n})\subseteq [n]
\]
be measurable, permutation-equivariant random index sets, interpreted respectively as a reconstruction set and an exception set.
Let $\mathcal{R}$ be a measurable map that associates with each finite tuple of points in $\mathcal{Z}$ a measurable subset of $\mathcal{Z}$.
Assume that, for every realization $z_{1:n}\in \mathcal{Z}^n$, 
\begin{enumerate}[label=(\alph*),leftmargin=1.8em]
\item if $i\notin E_n(z_{1:n})$, then $z_i\in \mathcal{R}(z_{B_n(z_{1:n})})$;
\item if $j\notin B_n(z_{1:n})\cup E_n(z_{1:n})$, then 
\[
\mathcal{R}(z_{B_n(z_{1:n})})=
\mathcal{R}\bigl((z^{-j}_{1:n})_{B_{n-1}(z^{-j}_{1:n})}\bigr),
\]
where $z^{-j}_{1:n}$ is the sample obtained by removing the $j$th point.
\end{enumerate}
Define the set predictor
\[
\Gamma(z_{1:m}):=\mathcal{R}(z_{B_m(z_{1:m})}). 
\]
Then
\[
\mathbb{P}\{Z_{m+1}\notin \Gamma(Z_{1:m})\}
\le
\frac{1}{m+1}\,\mathbb{E}\bigl[|B_{m+1}(Z_{1:m+1})\cup E_{m+1}(Z_{1:m+1})|\bigr].
\]
In particular, if $|B_n|\le \kappa$ and $|E_n|\le q$ almost surely for all $n$, then
\[
\mathbb{P}\{Z_{m+1}\notin \Gamma(Z_{1:m})\}\le \frac{\kappa+q}{m+1}.
\]
\end{lemma}

\begin{proof}
Let $T=(Z_1,\dots,Z_m,Z_{m+1})$.
Suppose that $Z_{m+1}\notin \Gamma(Z_{1:m})$.
We claim that $m+1\in B_{m+1}(T)\cup E_{m+1}(T)$.
If not, then by the stability property,
\[
\mathcal{R}(T_{B_{m+1}(T)})=
\mathcal{R}\bigl((Z_1,\dots,Z_m)_{B_m(Z_{1:m})}\bigr)=\Gamma(Z_{1:m}). 
\]
Since $m+1\notin E_{m+1}(T)$, the consistency property gives
\[
Z_{m+1}\in \mathcal{R}(T_{B_{m+1}(T)})=\Gamma(Z_{1:m}),
\]
which is a contradiction.
Therefore,
\[
\{Z_{m+1}\notin \Gamma(Z_{1:m})\}\subseteq \{m+1\in B_{m+1}(T)\cup E_{m+1}(T)\}. 
\]
By exchangeability and permutation equivariance,
\[
\mathbb{P}\{m+1\in B_{m+1}(T)\cup E_{m+1}(T)\}
=
\frac{1}{m+1}\,\mathbb{E}\bigl[|B_{m+1}(T)\cup E_{m+1}(T)|\bigr].
\]
This proves the first claim.
The second follows from $|B_{m+1}\cup E_{m+1}|\le |B_{m+1}|+|E_{m+1}|\le \kappa+q$. 
\end{proof}

\begin{remark}\rm
Lemma~\ref{lem:bridge} is used here as a bridge tool.
When $E_n=\varnothing$, it reduces to the stable-compression logic studied in the learning literature \citep{HannekeKontorovich2021,CampiGaratti2023}.
The point is not to present the lemma as a new learning-theoretic result, but to isolate the exchangeability mechanism that allows a conformal interpretation of control-oriented scenario programs with admissible exceptions.
\end{remark}

\section{Forward bridge for discarded scenarios} 

We now specialize Lemma~\ref{lem:bridge} to sample-and-discard algorithms.
\begin{theorem}[Forward bridge with discarded constraints]\label{thm:forward-discard}
Let $S=(\delta_1,\dots,\delta_m)$ be an i.i.d. 
scenario sample, let $\delta$ be an additional independent scenario, and suppose that Assumption~\ref{ass:stable-support} holds.
Then
\[
\mathbb{E}[V_r(S)]
\le
\frac{1}{m+1}\,\mathbb{E}\bigl[|D_r(S^+)\cup C_r(S^+)|\bigr],
\qquad S^+=(\delta_1,\dots,\delta_m,\delta).
\]
Consequently,
\[
\mathbb{E}[V_r(S)]
\le
\frac{r+\mathbb{E}[|C_r(S^+)|]}{m+1}. 
\]
If $|C_r(S)|\le \zeta$ almost surely for all samples, then 
\[
\mathbb{E}[V_r(S)]\le \frac{r+\zeta}{m+1}.
\]
In the fully supported convex case, one may take $\zeta=d$ and recover the law $(r+d)/(m+1)$.
\end{theorem}

\begin{proof}
Since $\delta$ is independent of $S$,
\[
\mathbb{E}[V_r(S)] = \mathbb{P}\{g(x_r(S),\delta)>0\},
\]
where the probability is taken jointly over $S$ and $\delta$.
Consider the augmented sample $S^+=(\delta_1,\dots,\delta_m,\delta)$. 
Suppose that $g(x_r(S),\delta)>0$.
We claim that the added index $m+1$ must belong to $D_r(S^+)\cup C_r(S^+)$.
If not, then by Assumption~\ref{ass:stable-support}(b), 
\[
x_r(S^+)=x_r(S).
\]
Moreover, $m+1\notin D_r(S^+)$ means that the fresh scenario is retained, so by Assumption~\ref{ass:stable-support}(a) it must be satisfied by the decision computed from $S^+$:
\[
g(x_r(S^+),\delta)\le 0.
\]
Using $x_r(S^+)=x_r(S)$ gives $g(x_r(S),\delta)\le 0$, which contradicts the starting assumption.
Therefore, 
\[
\{g(x_r(S),\delta)>0\}\subseteq \{m+1\in D_r(S^+)\cup C_r(S^+)\}.
\]
Taking probabilities and using exchangeability of the augmented sample,
\[
\mathbb{P}\{m+1\in D_r(S^+)\cup C_r(S^+)\}
=
\frac{1}{m+1}\,\mathbb{E}\bigl[|D_r(S^+)\cup C_r(S^+)|\bigr].
\]
This proves the first bound.
The second follows from $|D_r(S^+)|=r$ and $C_r(S^+)\subseteq [m+1]\setminus D_r(S^+)$.
The last claim is immediate when $|C_r(S)|\le \zeta$ almost surely.
\end{proof}

\begin{remark}[Interpretation for control design]\label{rem:interpretation}\rm
When $r=0$ and the retained problem is fully supported, Theorem~\ref{thm:forward-discard} reduces to the conformal-to-scenario direction derived by \citet{OSullivanRomaoMargellos2026}.
The resulting mean-risk law is classical in scenario optimization, so the novelty here is not a sharper numerical bound.
The novelty is the direct conformal-style derivation for sample discarding, together with the observation that discarded samples play the role of admissible exceptions in the exchangeability argument.
At the level of a sampled finite-horizon control-design problem, the theorem says that a fresh violating scenario must be either removed or  decision-shaping.
This interpretation is consistent with the compression view of scenario methods in control \citep{MargellosPrandiniLygeros2015} and with sample-and-remove scenario MPC formulations \citep{SchildbachFagianoFreiMorari2014}.
\end{remark}

\begin{example}[One-dimensional order-statistic predictor]\label{ex:order-statistic}\rm
Let $S=(R_1,\dots,R_m)$ be a sample of i.i.d.
real-valued observations from a continuous distribution with cumulative distribution function $F$.
For $r\in\{0,\dots,m-1\}$, consider the scalar optimization problem 
\[
\min_{q\in \mathbb{R}} q
\qquad\text{subject to}\qquad
R_i\le q \text{ for all but }r\text{ indices }i.
\]
Its unique solution is the order statistic $q_r(S)=R_{(m-r)}$, where $R_{(1)}\le \cdots \le R_{(m)}$.
The associated upper predictor is $\Gamma_r(S)=(-\infty,q_r(S)]$, and its conditional violation probability is
\[
V_r(S)=\mathbb{P}\{R>q_r(S)\mid S\}=1-F(q_r(S)).
\]
Because one retained point is enough to reconstruct the solution, Theorem~\ref{thm:forward-discard} gives
\[
\mathbb{E}[V_r(S)]\le \frac{r+1}{m+1}.
\]
In fact equality holds.
Since $F(q_r(S))$ is the $(m-r)$th order statistic of $m$ i.i.d. $\mathrm{Unif}[0,1]$ variables,
\[
F(q_r(S))\sim \mathrm{Beta}(m-r,r+1),
\qquad
V_r(S)=1-F(q_r(S))\sim \mathrm{Beta}(r+1,m-r).
\]
Therefore,
\[
\mathbb{E}[V_r(S)]=\frac{r+1}{m+1},
\]
and, for every $\epsilon\in(0,1)$,
\[
\mathbb{P}_S\{V_r(S)\le \epsilon\}
=
1-\sum_{i=0}^{r}\binom{m}{i}\epsilon^i(1-\epsilon)^{m-i}.
\]
This is exactly the scalar order-statistic law behind one-sided split-conformal calibration.
\end{example}

\section{A modular rule for multi-output and multi-stage risk allocation} 

The scalar example above gives an explicit calibration-conditional certificate.
We now show how to combine such certificates across several blocks.

\begin{proposition}[Modular calibration rule]\label{prop:multirisk}
Let $S$ be a calibration sample.
For $k=1,\dots,N$, let $\Gamma_k(S)\subseteq \mathcal{Z}$ be measurable set predictors and define
\[
V_k(S):=\mathbb{P}\{Z\notin \Gamma_k(S)\mid S\},
\]
where $Z$ is an independent test point in $\mathcal{Z}$.
Suppose that, for each $k$, there exist $\epsilon_k\in(0,1)$ and $\delta_k\in(0,1)$ such that 
\[
\mathbb{P}_S\{V_k(S)\le \epsilon_k\}\ge 1-\delta_k.
\]
Then the intersection predictor
\[
\Gamma_{\cap}(S):=\bigcap_{k=1}^N \Gamma_k(S)
\]
satisfies
\[
\mathbb{P}_S\!\left\{\mathbb{P}\{Z\notin \Gamma_{\cap}(S)\mid S\}\le \sum_{k=1}^N \epsilon_k\right\}
\ge
1-\sum_{k=1}^N \delta_k.
\]
\end{proposition}

\begin{proof}
Let $G_k:=\{V_k(S)\le \epsilon_k\}$ and $G:=\bigcap_{k=1}^N G_k$.
By the union bound,
\[
\mathbb{P}_S(G)\ge 1-\sum_{k=1}^N \delta_k.
\]
On the event $G$,
\[
\mathbb{P}\{Z\notin \Gamma_{\cap}(S)\mid S\}
=
\mathbb{P}\!\left(\bigcup_{k=1}^N \{Z\notin \Gamma_k(S)\}\,\middle|\,S\right)
\le
\sum_{k=1}^N V_k(S)
\le
\sum_{k=1}^N \epsilon_k.
\]
Taking probabilities over $S$ gives the claim. 
\end{proof}

\begin{remark}\rm
Proposition~\ref{prop:multirisk} is elementary---it is just a union bound---but its value is modularity.
It allows one to distribute a total risk budget across coordinates, subsystems, or time steps and then combine the resulting blocks into one certificate.
No independence across blocks is required.
In a receding-horizon controller, for example, each block could correspond to a future stage, an output component, or a separate safety specification.
\end{remark}

\begin{corollary}[Coordinate-wise upper boxes]\label{cor:boxes}
Let $Z=(Z^{(1)},\dots,Z^{(N)})$ be an $\mathbb{R}^N$-valued random vector, and let
\[
S=(Z_1,\dots,Z_m)
\]
be an i.i.d. calibration sample.
For each coordinate $k$, let
\[
Z_{(1)}^{(k)}\le \cdots \le Z_{(m)}^{(k)}
\]
be the order statistics of $Z_1^{(k)},\dots,Z_m^{(k)}$, choose an integer $r_k\in\{0,\dots,m-1\}$, and define
\[
q_k(S):=Z_{(m-r_k)}^{(k)},
\qquad
\Gamma_k(S):=\{z\in \mathbb{R}^N:\ z^{(k)}\le q_k(S)\}.
\]
Assume that each marginal distribution of $Z^{(k)}$ is continuous.
Then, for every choice of $\epsilon_k\in(0,1)$,
\[
\mathbb{P}_S\{V_k(S)\le \epsilon_k\}
=
1-\delta_k,
\qquad
\delta_k:=\sum_{i=0}^{r_k}\binom{m}{i}\epsilon_k^i(1-\epsilon_k)^{m-i},
\]
where $V_k(S)=\mathbb{P}\{Z\notin \Gamma_k(S)\mid S\}$.
Hence the rectangular predictor
\[
\Gamma_{\Box}(S):=\bigcap_{k=1}^N \Gamma_k(S)=\prod_{k=1}^N (-\infty,q_k(S)]
\]
satisfies
\[
\mathbb{P}_S\!\left\{\mathbb{P}\{Z\notin \Gamma_{\Box}(S)\mid S\}\le \sum_{k=1}^N \epsilon_k\right\}
\ge
1-\sum_{k=1}^N \delta_k.
\]
The result does not require the coordinates of $Z$ to be independent.
\end{corollary}

\begin{proof}
For each coordinate, Example~\ref{ex:order-statistic} gives
\[
V_k(S)\sim \mathrm{Beta}(r_k+1,m-r_k),
\]
and therefore
\[
\mathbb{P}_S\{V_k(S)\le \epsilon_k\}
=
1-\sum_{i=0}^{r_k}\binom{m}{i}\epsilon_k^i(1-\epsilon_k)^{m-i}.
\]
The last claim follows from Proposition~\ref{prop:multirisk}.
\end{proof}

\subsection{Sharper bounds under coordinate independence}

The modular rule in Proposition~\ref{prop:multirisk} and Corollary~\ref{cor:boxes} relies on a union bound, which guarantees validity even when the coordinates of $Z$ are highly dependent. However, if the prediction errors are known to be strictly independent across coordinates, we can replace the additive bounds with tighter multiplicative ones.

\begin{corollary}[Multiplicative coordinate-wise boxes]\label{cor:independent-boxes}
In the setting of Corollary~\ref{cor:boxes}, assume additionally that the coordinates $Z^{(1)},\dots,Z^{(N)}$ of the random vector $Z$ are mutually independent, and that this independence holds across the coordinates of the calibration sample $S$ as well.
Then, the rectangular predictor $\Gamma_{\Box}(S)$ satisfies 
\[
\mathbb{P}_S\!\left\{\mathbb{P}\{Z\notin \Gamma_{\Box}(S)\mid S\}\le 1 - \prod_{k=1}^N (1 - \epsilon_k)\right\}
\ge
\prod_{k=1}^N (1 - \delta_k).
\]
\end{corollary}

\begin{proof}
Given the calibration sample $S$, the test point $Z$ falls within the box $\Gamma_{\Box}(S)$ if and only if $Z^{(k)} \le q_k(S)$ for all $k=1,\dots,N$.
By the assumed independence of the test coordinates, the conditional probability of this joint event is the product of the marginal probabilities:
\[
\mathbb{P}\{Z\in \Gamma_{\Box}(S)\mid S\} 
= \prod_{k=1}^N \mathbb{P}\{Z^{(k)}\le q_k(S)\mid S\} 
= \prod_{k=1}^N (1 - V_k(S)).
\]
Thus, the joint conditional risk is $V_{\Box}(S) = 1 - \prod_{k=1}^N (1 - V_k(S))$.
Furthermore, because the calibration sample $S$ also has mutually independent coordinates, the order statistics $q_k(S)$—and consequently the marginal risks $V_k(S)$—are independent random variables across $k$.
Therefore, the joint confidence over the calibration phase is:
\[
\mathbb{P}_S\!\left\{\bigcap_{k=1}^N \{V_k(S)\le \epsilon_k\}\right\} 
= \prod_{k=1}^N \mathbb{P}_S\{V_k(S)\le \epsilon_k\} 
= \prod_{k=1}^N (1 - \delta_k).
\]
On the event $\bigcap_{k=1}^N \{V_k(S)\le \epsilon_k\}$, we have $V_{\Box}(S) \le 1 - \prod_{k=1}^N (1 - \epsilon_k)$.
The result follows. 
\end{proof}

\begin{remark}\rm
This multiplicative risk threshold is strictly smaller (tighter) than the summation bound $\sum_{k=1}^N \epsilon_k$.
Similarly, the multiplicative confidence is strictly larger (tighter) than the summation bound $1 - \sum_{k=1}^N \delta_k$.
While highly advantageous, this formulation is strictly reserved for applications where coordinate independence can be reliably verified.
\end{remark}

\begin{remark}[Residual scores and identified models]\label{rem:residual-scores}\rm
The order-statistic mechanism does not require the score to be the raw variable itself.
If $R_{j,k}$ is any continuous scalar score, then its $(m-r_k)$th order statistic has the same beta law as in Example~\ref{ex:order-statistic}.
In control and identification, the most natural choice is often the absolute residual score $R_{j,k}=|y_{j,k}-\hat y_{j,k}|$ at horizon $k$.
\end{remark}

\begin{example}[Calibrated multi-step output tube around an identified model]\label{ex:tube}\rm
Let $\Xi$ denote the current regressor together with any planned future inputs, let $Y\in \mathbb{R}^H$ be the corresponding future output vector, and let $Z=(\Xi,Y)$ denote a prediction task.
Assume that the nominal predictor $\hat y_{1:H}=f_{\mathrm{id}}(\Xi)$ has been fitted on training data independent of calibration.
Let
\[
S=\{(\xi_j,Y_j)\}_{j=1}^m
\]
be an exchangeable calibration sample of prediction tasks, and define the horizon-wise residual scores
\[
R_{j,k}:=|Y_j^{(k)}-\hat y_k(\xi_j)|,
\qquad k=1,\dots,H.
\]
Choose integers $r_k\in\{0,\dots,m-1\}$ and let $q_k(S)$ be the $(m-r_k)$th order statistic of $R_{1,k},\dots,R_{m,k}$.
By Remark~\ref{rem:residual-scores}, each $q_k(S)$ has the order-statistic law from Example~\ref{ex:order-statistic}.
Hence the blockwise predictors
\[
\Gamma_k(S):=\{(\xi,y):\ |y^{(k)}-\hat y_k(\xi)|\le q_k(S)\}
\]
come with explicit calibration-conditional certificates, and Proposition~\ref{prop:multirisk} yields the joint guarantee
\[
\mathbb{P}_S\!\left\{\mathbb{P}\{Y\notin \Gamma_{\mathrm{tube}}(S;\Xi)\mid S\}\le \sum_{k=1}^H \epsilon_k\right\}
\ge 1-\sum_{k=1}^H \delta_k,
\]
where the calibrated tube at a given regressor $\xi$ is
\[
\Gamma_{\mathrm{tube}}(S;\xi):=\{y\in \mathbb{R}^H:\ |y^{(k)}-\hat y_k(\xi)|\le q_k(S),\ k=1,\dots,H\}
=\prod_{k=1}^H [\hat y_k(\xi)-q_k(S),\,\hat y_k(\xi)+q_k(S)].
\]
For a scalar upper output constraint $y_k\le y_{\max}$, the tube can be used in predictive control by imposing the tightened nominal condition
\[
\hat y_k(\xi,u)\le y_{\max}-q_k(S),
\qquad k=1,\dots,H.
\]
The modular rule then allocates the total risk budget across prediction steps through the choice of $(\epsilon_k,\delta_k)$.
\end{example}

\begin{remark}[Conditional guarantees and impossibility results]\rm
The certificate in Corollary~\ref{cor:boxes} and Example~\ref{ex:tube} should not be confused with a generic conditional guarantee based only on observed support size.
In scenario optimization, such generic conditional claims are impossible without additional information \citep{GarattiCampi2023Conditional}.
Here the situation is different: each block has an explicit conditional law inherited from an order-statistic score, and Proposition~\ref{prop:multirisk} merely combines those explicit blockwise certificates.
\end{remark}

\section{Numerical illustrations}\label{sec:numerics} 

This section illustrates Example~\ref{ex:tube} on a simple identified-predictor setting.
We compare several horizon-wise allocations of the same total risk budget and then embed the resulting calibrated tubes in a simple planning problem.

\subsection{Risk allocation across the prediction horizon}

We use the nonlinear stochastic plant
\[
x_{t+1}=0.78x_t+0.35u_t+0.12x_tu_t+w_t,
\qquad y_t=x_t,
\]
where $w_t\sim \mathcal{N}(0,0.08^2)$.
The nominal predictor is the fixed linear surrogate
\[
\hat y_{t+1}=\hat a y_t+\hat b u_t,
\qquad
\hat a=0.7799,\quad \hat b=0.3491,
\]
which is treated as the output of an identification step performed on data independent of the calibration and test tasks.
Each prediction task consists of an initial output $y_0\sim \mathcal{N}(0,1)$ and a planned input sequence $u_0,\dots,u_3$, i.i.d.
uniform on $[-1,1]$. 
The nominal model is iterated over four steps to produce $\hat y_{1:4}$.
A calibration set contains $m=120$ independent tasks, and the score at horizon $k$ is the absolute residual
\[
R_{j,k}=|y_{j,k}-\hat y_{j,k}|.
\]

We compare three modular allocations, all using the same total risk budget $\sum_{k=1}^4 \epsilon_k=0.22$ and the same total rank sum $\sum_{k=1}^4 r_k=6$:
\[
\text{increasing risk: } (r_1,r_2,r_3,r_4)=(0,1,2,3),\quad
(\epsilon_1,\epsilon_2,\epsilon_3,\epsilon_4)=(0.04,0.05,0.06,0.07),
\]
\[
\text{uniform risk: } (r_1,r_2,r_3,r_4)=(1,1,2,2),\quad
(\epsilon_1,\epsilon_2,\epsilon_3,\epsilon_4)=(0.055,0.055,0.055,0.055),
\]
\[
\text{decreasing risk: } (r_1,r_2,r_3,r_4)=(3,2,1,0),\quad
(\epsilon_1,\epsilon_2,\epsilon_3,\epsilon_4)=(0.07,0.06,0.05,0.04).
\]
The terminology refers to the intended stage-wise allocation of allowable risk along the horizon.
For each allocation we generated $1{,}000$ independent calibration sets.
For every calibration set, the stage-wise risks $V_k(S)$ and the joint trajectory risk
\[
V_{\mathrm{traj}}(S):=\mathbb{P}\{Y_{1:4}\notin \Gamma_{\mathrm{tube}}(S;\Xi)\mid S\}
\]
were approximated by $5{,}000$ fresh test tasks.
Figure~\ref{fig:risk-allocation} shows how the modular rule redistributes conservatism across the horizon.
As expected, the increasing-risk profile yields the smallest margins at later steps, whereas the decreasing-risk profile shifts conservatism toward the end of the horizon.
The empirical stage-wise risks follow the opposite pattern.
Table~\ref{tab:allocations} summarizes the resulting certificates and trajectory-risk statistics.
The mean trajectory risk remains close to $0.06$ in all three cases, but the allocation of that risk across stages changes substantially.
This is exactly the design freedom offered by Proposition~\ref{prop:multirisk}: one can preserve an overall certificate while reshaping local margins across time.
For reference, we also compared these modular tubes with a joint baseline obtained by calibrating the supremum residual across the horizon to the same total risk budget.
In our simulations that baseline produced a slightly tighter but uniform tube, whereas the modular construction was marginally more conservative and, crucially, allowed deliberate stage-dependent shaping of the margins.

\subsubsection*{Evaluating the Multiplicative Savings under Independence}

To quantify the theoretical savings of the multiplicative rule from Corollary~\ref{cor:independent-boxes}, suppose for a moment that the prediction errors across the $H=4$ horizon steps were statistically independent.
In the ``uniform risk'' allocation, the additive rule relies on a stage-wise risk of $\epsilon_k = 0.055$, yielding a total trajectory risk bound of $\sum_{k=1}^4 \epsilon_k = 0.22$.
Under the assumption of independence, the exact same stage-wise margins would yield a tighter trajectory risk of:
\[
V_{\mathrm{traj}}^{\mathrm{mult}} = 1 - \prod_{k=1}^4 (1 - 0.055) = 1 - (0.945)^4 \approx 0.2026.
\]
This represents a purely mathematical reduction of about 1.74\% in the certified upper bound of the trajectory risk, without changing the physical tube at all.
Viewed from a design perspective, if the control engineer's target budget is strictly $0.22$, the multiplicative rule allows for a uniformly larger stage-wise risk $\epsilon_k$.
Solving $1 - (1 - \epsilon_k)^4 = 0.22$ gives $\epsilon_k \approx 0.061$ for each step (compared to $0.055$ in the additive case).
This increased local budget allows the controller to discard more samples (higher $r_k$) or select smaller order statistics $q_k(S)$, resulting in a uniformly tighter and less conservative tube while maintaining the exact same trajectory-level safety certificate.

\begin{figure}[t]
\centering
\includegraphics[width=0.98\linewidth]{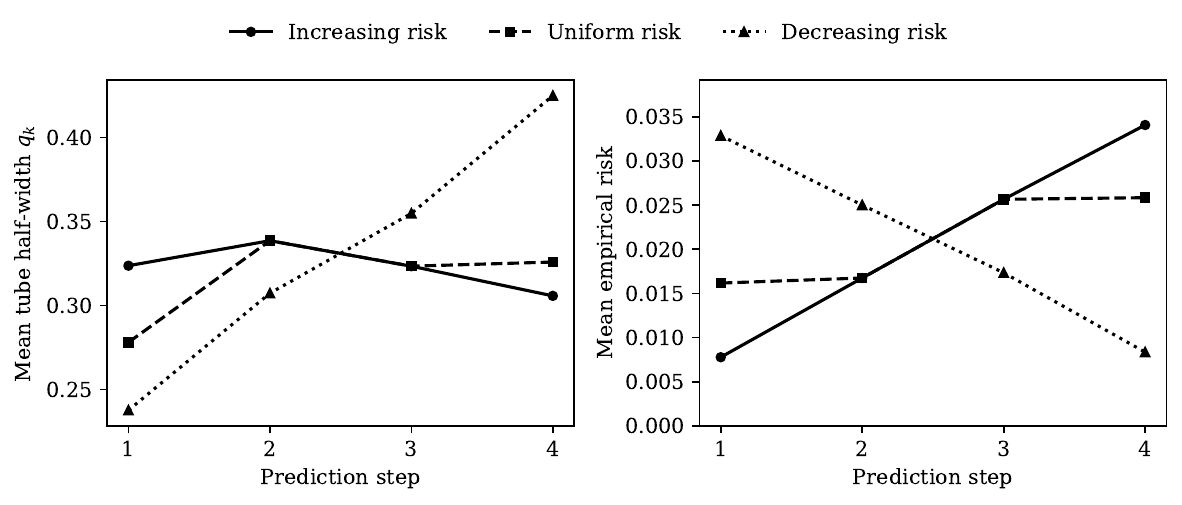}
\caption{Effect of modular allocation across the horizon. Left: mean tube half-widths $q_k$. Right: mean empirical stage-wise risks $V_k(S)$.
The three allocations use the same total risk budget and the same total rank sum, but redistribute conservatism across prediction steps.}
\label{fig:risk-allocation}
\end{figure}

\begin{table}[t]
\centering
\scriptsize
\caption{Alternative horizon-wise allocations for the calibrated four-step tube.
All three choices use $\sum_k \epsilon_k=0.22$ and $\sum_k r_k=6$.}
\label{tab:allocations}
\begin{tabular}{lcccccc}
\toprule
allocation & $(r_1,\ldots,r_4)$ & $(\epsilon_1,\ldots,\epsilon_4)$ & certificate $1-\sum_k\delta_k$ & mean $V_{\mathrm{traj}}$ & $Q_{0.90}(V_{\mathrm{traj}})$ & $Q_{0.99}(V_{\mathrm{traj}})$ \\
\midrule
increasing risk & $(0,1,2,3)$ & $(0.04,0.05,0.06,0.07)$ & 0.9264 & 0.0618 & 0.0932 & 0.1288 \\
uniform risk    & $(1,1,2,2)$ & $(0.055,0.055,0.055,0.055)$ & 0.9095 & 0.0613 & 0.0930 & 0.1270 \\
decreasing risk & $(3,2,1,0)$ & $(0.07,0.06,0.05,0.04)$ & 0.9264 & 0.0626 & 0.0940 & 0.1268 \\
\bottomrule
\end{tabular}
\end{table}

\subsection{From calibrated tubes to constraint tightening}

To make the control interpretation explicit, we now translate each calibrated tube into a simple finite-horizon planning problem.
Fix the initial condition $y_0=0.1$ and the upper output limit $y_{\max}=0.7$.
We consider the scalar optimization problem
\[
\max_{0\le u\le 1} u
\qquad\text{subject to}\qquad
\hat y_k(y_0,u)\le y_{\max}-q_k(S),\quad k=1,\dots,4,
\]
where a constant input $u$ is applied over the whole horizon and $q_k(S)$ is the calibrated margin obtained from one of the three allocations above.
This is a simple surrogate of a predictive-control step: the decision variable is scalar, but the role of the calibrated margins is exactly the same as in tube-based constraint tightening.
For each calibration set we solved this problem and obtained an admissible input $u^\star(S)$.
We then evaluated the true nonlinear plant under this chosen input over $4{,}000$ fresh noise realizations.
Table~\ref{tab:planning} reports the resulting performance and safety statistics.
The increasing-risk allocation is the least conservative: it yields the largest mean admissible input, $0.3491$, and the largest mean terminal output, $0.4190$, but also the largest mean probability of violating the true output limit, $0.0174$.
The decreasing-risk allocation reverses this trade-off: by allocating more conservatism to the later prediction steps, it reduces the mean admissible input to $0.2373$ and lowers the mean violation probability to $0.0025$.
The uniform allocation lies in between.

These numbers show that the modular rule is not merely a bookkeeping device.
Even when the total risk budget is the same, different stage-wise allocations induce materially different tightened constraints and therefore materially different control actions.
Figure~\ref{fig:control-tube} shows one representative tightened tube for the increasing-risk allocation together with several true trajectories generated under the resulting constant input.

\begin{table}[t]
\centering
\scriptsize
\caption{Simple planning problem with calibrated constraint tightening. The controller chooses the largest constant input $u^\star\in[0,1]$ such that $\hat y_k(y_0,u^\star)\le y_{\max}-q_k(S)$ for $k=1,\dots,4$, with $y_0=0.1$ and $y_{\max}=0.7$.}
\label{tab:planning}
\begin{tabular}{lcccccc}
\toprule
allocation & mean $u^\star$ & $Q_{0.10}(u^\star)$ & $Q_{0.90}(u^\star)$ & mean viol.
prob. & $Q_{0.90}(\text{viol. prob.})$ & mean terminal output \\
\midrule
increasing risk & 0.3491 & 0.2992 & 0.3925 & 0.0174 & 0.0353 & 0.4190 \\
uniform risk    & 0.3331 & 0.2801 & 0.3778 & 0.0123 & 0.0260 & 0.4004 \\
decreasing risk & 0.2373 & 0.1219 & 0.3244 & 0.0025 & 0.0068 & 0.2923 \\
\bottomrule
\end{tabular}
\end{table}

\begin{figure}[t]
\centering
\includegraphics[width=0.84\linewidth]{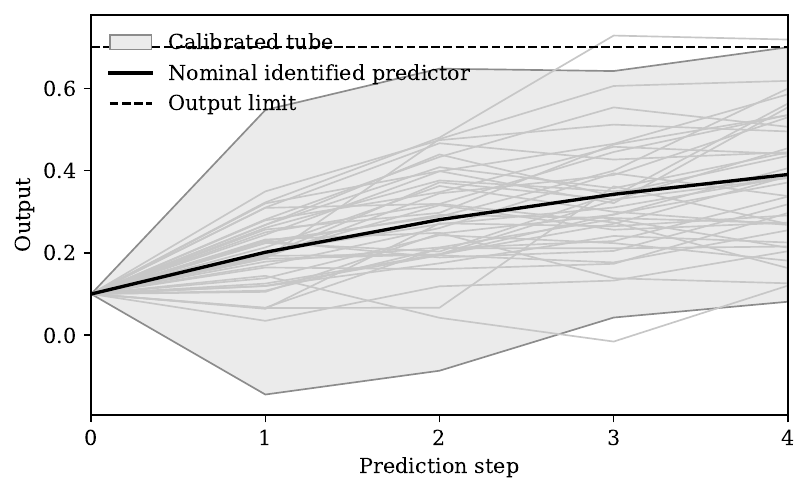}
\caption{Representative tightened tube for the increasing-risk allocation in the simple planning problem.
The band is computed from one calibration set, the solid line is the nominal identified trajectory under the selected constant input $u^\star$, the gray lines are true trajectories of the nonlinear plant, and the dashed line is the output limit.}
\label{fig:control-tube}
\end{figure}

\FloatBarrier

\section{Conclusions}

This note revisited the connection between conformal prediction and scenario optimization through a dedicated systems-and-control lens.
Our primary theoretical result is a forward bridge for sample discarding.
Specifically, if a feasible sample-and-discard algorithm admits a stable retained reconstruction set, the same exchangeability logic used in conformal calibration yields a mean-risk bound defined entirely in terms of discarded and reconstruction samples.
For the fully supported convex case, this directly recovers the standard $(r+d)/(m+1)$ law.
Our second main result is a modular calibration rule designed to combine several blockwise certificates.
While the rule itself is elementary, it performs a crucial function: making risk allocation explicit across outputs, constraints, or prediction steps.
This is a highly practical tool whenever an identified or learned predictor must be converted into margins that a downstream controller can actually use.
As our numerical illustrations demonstrated, this modularity is operationally meaningful.
Even when constrained by the exact same total risk budget, different stage-wise allocations actively produce different tubes, altered constraint tightenings, and shifted performance-safety trade-offs in a planning problem.
Our contribution has clear limitations. The forward bridge yields an expectation bound rather than a general, sharp arbitrary-$r$ PAC tail law, and the modular composition rule relies fundamentally on a union bound.
These limitations immediately point to two natural next steps for future research.
First, exploring sharper joint constructions that actively exploit dependence across blocks to reduce conservatism.
Second, investigating closed-loop extensions where calibration, prediction, and control seamlessly interact over time.

\end{document}